\documentclass[12pt, reqno]{amsart}
\usepackage{amsthm}
\usepackage{amssymb}
\usepackage{amsmath}
\usepackage{epic, eepic}

\def\1{\hbox{1\kern-.35em\hbox{1}}}

\numberwithin{equation}{section}

\newtheorem{theorem}{Theorem}[section]
\newtheorem*{theorem*}{Theorem}
\newtheorem{lemma}{Lemma}[section]
\newtheorem{proposition}{Proposition}[section]
\newtheorem*{proposition*}{Proposition}

\newtheorem{remark}{Remark}[section]

\newcommand{\Z}{{\mathbb Z}}

\newcommand{\R}{{\mathbb R}}
\newcommand{\C}{{\mathbb C}}

\newcommand{\fg}{{\mathfrak g}}

\newcommand{\fp}{{\mathfrak p}}

\newcommand{\U}{{\rm U}}

\newcommand{\cE}{{\mathcal E}}
\newcommand{\cF}{{\mathcal F}}
\newcommand{\cH}{{\mathcal H}}

\newcommand{\cU}{{\mathcal U}}
\newcommand{\cN}{{\mathcal N}}
\newcommand{\cV}{{\mathcal V}}

\begin{document}

\title[Quantum Kepler problems on superspaces]
{Orthosymplectic Lie superalgebras in \\ superspace analogues of quantum Kepler problems}
\author{R. B. Zhang}
\address{School of Mathematics and Statistics, University of Sydney,
Australia} \email{rzhang@maths.usyd.edu.au}

\begin{abstract}
A Schr\"odinger type equation on the superspace $\R^{D|2n}$ is studied,
which involves a potential inversely proportional to the negative of
the $osp(D|2n)$ invariant ``distance" away from the origin. An
$osp(2, D+1|2n)$ dynamical supersymmetry for the system is
explicitly constructed, and the bound states of the system are shown
to form an irreducible highest weight module for this superalgebra.
A thorough understanding of the structure of the irreducible module is obtained.
This in particular enables the determination of the energy
eigenvalues and the corresponding eigenspaces as well as their
respective dimensions.
%\noindent{\bf PACS2006}: 02.20.Sv; 02.20.Hj.
%\noindent{\bf Key words}: orthosymplectic Lie superalgebras; highest weight representations; Kepler problem.
\end{abstract}

\maketitle

%\tableofcontents

\section{Introduction}
The quantum Kepler problem and its analogues in higher dimensions
are a series of soluble quantum mechanical systems with $-\frac{1}{r}$
potentials. The simplest of such systems is the hydrogen atom. It
was discovered in the 60s by McIntosh and Cisneros \cite{MC} and Zwanziger \cite{Z}
independently that the quantum system describing the
hydrogen atom remained soluble when coupled to a magnetic charge.
Since then the quantum Kepler problem has been generalised to
include couplings to nonabelian magnetic monopoles in 5-dimensions in \cite {I} and in arbitrary
dimensions in \cite{M1}. Such generalisations are referred to as the
generalised MICZ-Kepler problem in the literature.

It has long been known that the
original quantum Kepler problem had an $so(2, 4)$ dynamical
symmetry. In \cite{BB}, Barut and Bornzin showed that the dynamical
symmetry survived when the system was under the influence of a
magnetic charge. They used the dynamical symmetry to give a beautiful solution
of the problem. In a recent joint publication \cite{MZ} with Meng,
we demonstrated that the
generalised MICZ-Kepler problem in odd dimension $D$ has an $so(2,
D+1)$ dynamical symmetry, and gave a solution to the problem by
algebraic means using a particular irreducible unitary highest
weight representation of the dynamical symmetry. This work has also
been extended to even dimensions (but for restricted classes of
magnetic monopoles) in \cite{M2}.

A natural problem is to introduce supersymmetries into the Kepler
problem and its generalisations and to study the resulting
supersymmetric quantum mechanical systems. The $\cN=2$
supersymmetric case was studied in \cite{KLPW} in arbitrary
dimensions but without magnetic monopoles. An $so(D+1)$ dynamical
symmetry remained in this case, which helped to obtain the bound states spectrum and the
multiplicities of the eigenvalues.

In this paper we investigate the Schr\"{o}dinger equation on the
superspace $\R^{D|2n}$ involving a potential inversely proportional
to the negative of the $osp(D|2n)$ invariant ``distance" away from
the origin (see equation \eqref{Hamiltonian}). We shall refer to the
study of this eigenvalue problem as the quantum Kepler problem on
the superspace $\R^{D|2n}$ without magnetic monopoles. As we shall see,
the system is integrable, and will be solved by algebraic means.
It will be very interesting to extend the study to include
couplings to magnetic monopoles.

Schr\"{o}dinger equations on superspaces were studied by Delbourgo
in the late 80s \cite{D1} as a method to incorporate spin. This
developed into a fruitful programme (see, e.g., \cite{DJW} and
references therein) on using supermanifolds to describe gauge
symmetries and also to explain internal degrees of freedom of
elementary particles (see \cite{D2, D3} for recent developments). We
hope that the present work and further studies will provide useful
mathematical information for the programme of Delbourgo and
co-workers.

Our primary interest in this paper is the integrability of the
quantum Kepler problem on superspace without magnetic monopoles, and
also the relevant representations of its dynamical supersymmetry. We
shall give the precise definition of the quantum Kepler problem on
the superspace $\R^{D|2n}$ in Section \ref{Keplerproblem}, and
reformulate the problem algebraically following ideas of Barut and
Bornzin \cite{BB}. We show in Theorem \ref{osp-dynamicalsymmetry}
that the problem has an $osp(2, D+1|2n)$ dynamical supersymmetry,
where the generators of the orthosymplectic Lie superalgebra are
constructed explicitly. The bound states are shown to form an
infinite dimensional irreducible highest weight module of this Lie
superalgebra (Theorem \ref{main-rep}), and a thorough understanding
of the structure of this module is also given. Using this
information we obtain the bound state spectrum and also the
corresponding eigenspaces in Theorem \ref{main-spec}.

It is quite remarkable that the Kepler problem
remains integrable when generalised to superspaces
(see Remarks \ref{comments1} and \ref{comments2}). Furthermore,
the appearance of the dynamical supersymmetry and the way in which
its representation theory enables us to solve the problem all appear to be
quite fascinating from the point of view of the theory of
Lie superalgebras. It is well known that even the finite dimensional
representations of orthosymplectic
Lie superalgebras are extremely hard to study and very little is known about them.
Thus it is a nice surprise that the infinite dimensional irreducible representation of
$osp(2, D+1|2n)$ appearing in the problem can be understood for all $D>2n+1$. Therefore, results in this paper should be
of interest to the representation theory of Lie superalgebras as well.

\section{Quantum Kepler problem on superspace}\label{Keplerproblem}
In this section we introduce the quantum Kepler problems on
superspaces, and also give an algebraic formulation for the problems
following the strategy of Barut and Bornzin \cite{BB}.

\subsection{Generalities}\label{subsect-generalities}

Let $\R^{D|2n}$ denote the superspace with $D$ even dimensions and
$2n$ odd dimensions. Denote by $X^a$ with $a=1, 2, \dots, D+2n$ the
coordinate of the superspace, where $X^a$ is even if $a\le D$ and
odd if $a>D$. We shall assign to $\R^{D|2n}$ the  metric $\eta =
(\eta_{a b}) =
\begin{pmatrix} I_D & 0 & 0\\ 0 & 0 &-I_n\\ 0 & I_n &
0\end{pmatrix}$, where $I_D$ and $I_n$ are the identity matrices of
sizes $D\times D$ and $n\times n$ respectively. By a function on
$\R^{D|2n}$ we shall mean a map from $\R^D$ to the complex Grassmann
algebra $\Lambda_{2n}$ generated by the odd coordinates $X^{\nu}$ with $\nu=D+1,
D+2, \dots, D+2n$. Denote $\partial_a =\frac{\partial}{\partial
X^a}$ which acts on functions from the left. Let $X_a =
\sum_b \eta_{a b} X^b$ and $\partial^a = \sum_b \eta^{a b}
\partial_b$, where $\eta^{a b}$ are the entries of $\eta^{-1}$.
Set $\Delta = \sum_a \partial^a \partial_a$.

Given a function $\cV(X)$ on $\R^{D|2n}$ which is assumed
to be even in the Grassmann variables, we introduce the operator
\[ H = -\frac{1}{2}\Delta + \cV(X), \]
which will be referred to as a quantum Hamiltonian operator.
Our broad aim is to investigate the eigenvalue problem for
the quantum Hamiltonian operator, that is, to solve the Schr\"{o}dinger equation
\begin{eqnarray} \label{Schroedinger}
H\Psi = \cE \Psi,
\end{eqnarray}
where the eigenvalue $\cE$ is required to belong to $\R$ (thus $\cV(X)$ has to be even).
We shall be particularly interested in systems of the form of
\eqref{Schroedinger} which are integrable.

For each given $\cV(X)$, we shall need to specify the class of
functions on $\R^{D|2n}$, to which the solutions of the
Schr\"{o}dinger equation belong. For the potential corresponding to
the Kepler problem, this will be discussed in some detail in the
next subsection. Here we merely point out that the eigenfunction
$\Psi$ is a polynomial in the odd coordinates with coefficients
being complex valued functions on $\R^D$, which will be referred to
as coefficient functions.

\begin{remark}\label{comments1}
The Schr\"{o}dinger equation on $\R^{D|2n}$ is equivalent to a system of
partial differential equations on $\R^D$ for the coefficient functions.
\end{remark}

Note that since the Hamiltonian operator is even,
the Schr\"{o}dinger equation separates into two independent equations for
the even and odd parts of $\Psi$ respectively.

\subsection{Quantum Kepler problem on superspace}\label{Kepler}

Let us now introduce the quantum Hamiltonian operator which we shall
study in this paper. Let $R = \left(\sum_a X^a X_a
\right)^{\frac{1}{2}}$, which is only defined away from the origin
of $\R^D$, and should be interpreted as a polynomial in the odd
coordinates. More precisely, let $r^2=\sum_{i, j=1}^D X^i\eta_{ i
j}X^j$ and $\Theta^2 = \sum_{\mu, \nu=D+1}^{D+2n} X^\mu\eta_{ \mu
\nu}X^\nu$. Then $R= r \sqrt{1+\frac{\Theta^2}{r^2}}$, where
$\sqrt{1+\frac{\Theta^2}{r^2}}$ should be understood as a Taylor
expansion in $\frac{\Theta^2}{r^2}$. Since the odd coordinates are
Grassmannian, the expansion terminates at order $n$ in
$\frac{\Theta^2}{r^2}$, and we have a polynomial in the odd
coordinates. We have
\begin{eqnarray}\label{dR}
\partial_a(R) = \frac{X_a}{R}, \quad \Delta (R^2) = 2 d, \quad \text{where $d=D-2n$}.
\end{eqnarray}

We shall take the following quantum Hamiltonian operator
\begin{eqnarray} \label{Hamiltonian}
H = -\frac{1}{2}\Delta - \frac{1}{R}.
\end{eqnarray}
Our purpose is to determine the spectrum of $H$ and the
corresponding eigenvectors. In the remainder of this paper, we shall
only consider the Schr\"odinger equation \eqref{Schroedinger} with
this quantum Hamiltonian.

Let us now specify the class of functions on $\R^{D|2n}$ to which
the eigenfunctions belong. We should mention that the so-called
superanalysis (analysis of functions on superspace) is yet to
develop into a coherent theory. It will take us too far astray to investigate
superanalysis in any depth here, as we shall adopt
an algebraic approach to the quantum Kepler problem on superspace, which
by-passes many analytic issues. The rather superfluous discussion below
on functions on $\R^{D|2n}$ suffices for us to get by.

The Grassmann algebra $\Lambda_{2n}$ generated by the odd
coordinates  is $\Z_+$-graded with
$X^\mu$  ($1+D\le \mu\le 2n+D$) having degree $1$. Let $\zeta_s$ ($0\le s\le 2^{2n}-1$) be a
homogeneous basis of $\Lambda_{2n}$ consisting of products of the
odd coordinates, and denote by $deg(\zeta_s)$ the degree of
$\zeta_s$. We order the basis elements in such a way that
$deg(\zeta_s)\le deg(\zeta_{s+1})$.

Introduce the conjugate linear algebra automorphism $\bar{\ }:
\Lambda_{2n}\longrightarrow \Lambda_{2n}$ defined by
\[ \overline{X^\mu} = X_\mu, \quad \mu=D+1, D+2, \dots, D+2n. \]
As usual, conjugate linear means that for any
$\lambda=c_1\lambda_1+c_2\lambda_2$ with $\lambda_1,
\lambda_2\in\Lambda_{2n}$ and $c_1, c_2\in\C$, $\bar\lambda =
\bar{c}_1\bar{\lambda}_1+\bar{c}_2\bar{\lambda}_2$, where $\bar c_1$
and $\bar c_2$ are the complex conjugates of $c_1$ and $c_1$. Also,
being an algebra automorphism, the map $\bar{\ }$ obeys the rule
$\overline{\lambda_1\lambda_2}=\bar\lambda_1\bar\lambda_2$. Note
that $\bar\zeta_s \zeta_s\ne 0$  if $deg(\zeta_s)\le n$, but
$\bar\zeta_t\zeta_t=0$ if $deg(\zeta_t)>n$.

The conjugate linear automorphism on $\Lambda_{2n}$ extends to the
superalgebra of functions on $\R^{D|2n}$ in a natural way, and we
shall still denote the resulting map by $\bar{\ }$. More explicitly,
write a function $\Psi$ as $\Psi=\sum_s\zeta_s \psi_s$ where the
$\psi_s$ are complex valued functions on $\R^D$. Then $\bar\Psi =
\sum_s \bar\zeta_s \bar\psi_s$, where $\bar\psi_s$ is the usual
complex conjugate of $\psi_s$. For any two functions $\Phi$ and
$\Psi$ on $\R^{D|2n}$, we let $ \langle\Phi\mid \Psi\rangle =
\int_{\R^D}{\overline{\Phi}} \Psi $ if the integral over $\R^D$
exists (thus lies in $\Lambda_{2n}$).

Let $\cF$ denote the set of functions on $\R^{D|2n}$ such that
for every $\Psi\in \cF$
\begin{enumerate}
\item the integral $\langle\Psi\mid \Psi\rangle =\int_{\R^D}{\overline{\Psi}} \Psi$ exists; and
\item
the coefficient functions of $\Psi$ are twice differentiable on
$\R^D\backslash \{0\}$.
\end{enumerate}
Then for any $\Psi=\sum_s\zeta_s \psi_s\in\cF$, the coefficient
function $\psi_0$, which will be called the body of $\Psi$, must be
square integrable in the usual sense. However note that all
functions $\Psi$ satisfying $\psi_s=0$ for all $s\le 2^n-1$ belong
to $\cF$.

\begin{remark}
Even though our definition of $\cF$ imposes conditions on
all the coefficient functions $\psi_s$ with $s\le 2^n-1$,
we have to some extend followed the type of thinking common in physics that
appropriate conditions only need to be imposed on the body (if the body
is nonzero) of a function and then the supersymmetry of the physical
problem will determine properties of the function as a whole.
\end{remark}

\begin{remark}
When defining $\cF$, one may be tempted to use the more
``natural" valuation $\int_{\R^{D|2n}}{\overline{\Psi}} \Psi$ instead, where
the integration over the odd coordinates is given by the Berezin
integral. However, this does not make the body of $\Psi$ square
integrable.
\end{remark}

We require that a solution of the Schr\"odinger equation
\eqref{Schroedinger} with the quantum Hamiltonian operator
\eqref{Hamiltonian} belongs to the set $\cF$. The complex vector
space spanned by all the solutions is evidently stable under the
action of the quantum Hamiltonian operator, and a major aim is to
understand this vector space. It is also this vector space which
the dynamical supersymmetry algebra acts on. As we shall see later,
the space in fact forms an irreducible module
over the dynamical supersymmetry algebra.

\medskip

Note that $\overline{\Theta^2} = \Theta^2$, and hence $\bar R=R$. Thus for any function $\Psi$, we have
\[ \overline{H\Psi} = H\overline\Psi. \]
If $\Psi$ is a solution of the Sch\"odinger equation satisfying
$\overline\Psi=\Psi$, then the corresponding eigenvalue will
necessarily be real. Thus we may regard the
quantum Hamiltonian operator as Hermitian in a generalised sense.

We shall further restrict ourselves to considering only the bound
states, that is, the eigenvectors of $H$ associated with negative
eigenvalues. In order to have bound state solutions in $\cF$, we
need to impose the following condition on the superspace:
$$D>2n+1.$$ In this case, the ground state eigenvalue $\cE_0$ and
the associated eigenvector $\Psi_0$ of the system are given by
\[ \cE_0 = -\frac{1}{2}\left(\frac{2}{d-1}\right)^2,
\quad \Psi_0 = \exp\left(- \frac{2}{d-1} R\right).\] Note that
$\int_{\R^D}\bar\Psi_0 \Psi_0$ exists thus $\Psi_0\in\cF$.  It is
also easy to check that $\Psi_0$ indeed satisfies the Schr\"odinger
equation with the energy eigenvalue $\cE_0$ (also see equation
\eqref{radial-part} and discussions after).

\subsection{Algebraic formulation}

We shall solve the quantum Kepler problem algebraically by using the
representation theory of a dynamical supersymmetry algebra,
which will be constructed later. For
this purpose we need to reformulate the Schr\"odinger equation \eqref{Schroedinger}
algebraically for the quantum Hamiltonian \eqref{Hamiltonian}.

It is well known that the following differential operators
\begin{eqnarray}\label{osp-small}
J_{a b} = X_a \partial_b - (-1)^{[a][b]} X_b \partial_a
\end{eqnarray}
form the orthosymplectic Lie superalgebra $osp(D|2n)$ with the
commutation relations (see, e.g., \cite[(40)]{JG})
\begin{eqnarray}\label{osp}
[J_{a b}, J_{c d}] = \eta_{c b} J_{a d} + (-1)^{[a]([b]+[c])}\eta_{d a} J_{b c}
- (-1)^{[c][d]}\eta_{d b} J_{a c} - (-1)^{[a][b]}\eta_{c a} J_{b d}.
\end{eqnarray}
Here $[a]=0$ if $a\le D$ and $[a]=1$ otherwise. As customary, $[A,
B]$ represents the usual commutator unless both $A$ and $B$ are odd
and in that case $[A, B]=A B + B A$. One can easily check that
\begin{eqnarray}\label{RDelta-inv}
[J_{a b}, \Delta] =0, \quad [J_{a b}, R]=0.
\end{eqnarray}
Therefore, $[J_{a b}, H]=0$ for all $a, b$,  and the system has an $osp(D|2n)$ symmetry. Let
\[
E=\sum_{a=1}^{D+2n} X^a \partial_a, \quad T = E + \frac{d-1}{2},
\]
where $E$ is the Euler operator. The following lemma is a
generalisation of a result in \cite[Appendix]{BB} to the superspace
setting.
\begin{lemma}\label{so(2,1)}
Define the following operators
\begin{eqnarray}
J_{-2}=T, \quad J_{-1} = \frac{i}{2} R \left(-\Delta -1\right), \quad
J_{0} =\frac{i}{2} R \left(-\Delta +1\right).
\end{eqnarray}
\begin{enumerate}
\item The operators satisfy
the commutation relations of the Lie algebra $so(2, 1)$:
\begin{eqnarray} \label{so-commut}
[J_{-1}, J_{0} ] =  J_{-2}, \quad [J_{-2}, J_{-1} ]
 = - J_{0},  \quad  [J_{0}, J_{-2} ] =  J_{-1}.
\end{eqnarray}
\item This Lie algebra commutes with the Lie superalgebra $osp(D|2n)$
spanned by the operators $J_{a b}$:
\[ [J_{a b}, J_0] = [J_{a b}, J_{-1}] = [J_{a b}, J_{-2}]=0, \quad \forall a, b.\]
\end{enumerate}
\end{lemma}
\begin{proof} Note that the Euler operator satisfies $[E, X^a]=X^a$ and $[E, \partial_a]=-\partial_a$.
This immediately leads to $[E, J_{a b}]=0$.
Now part (2) easily follows from this commutation relation and also the commutation
relations \eqref{RDelta-inv}. To prove the first relation in \eqref{so-commut},
note that
$[J_{-1}, J_{0} ] = \frac{1}{2}R[\Delta, R]$. Using
\begin{eqnarray}\label{DeltaR}
[\Delta, R] = \frac{2}{R}\left(\frac{d-1}{2} + E\right),
\end{eqnarray}
we obtain the desired result. The other two relations are easily proven
by using properties of the Euler operator $E$.
\end{proof}

\begin{theorem}\label{formulation}
For the bound states (with $\cE<0$), let
\[
\Phi = g \Psi,  \quad \text{where\ \ } g=\exp\left(-T \ln\sqrt{-2\cE} \right).
\]
Then the Schr\"{o}dinger equation
\eqref{Schroedinger} is equivalent to
\begin{eqnarray}\label{Schroedinger-1}
h_{0} \Phi =-\frac{1}{\sqrt{-2\cE}}\Phi \quad \text{where \ } h_0=iJ_{0}.
\end{eqnarray}
\end{theorem}
\begin{proof}
Let $\Theta= R(H-\cE)$. The Schr\"{o}dinger equation can be re-written
as $g\Theta\Psi=0$. Using the following relations
\[g R g^{-1} =\frac{1}{\sqrt{-2\cE}} R, \quad g \Delta g^{-1}
=-2{\cE}\Delta  \]
in the equation we obtain \eqref{Schroedinger-1}.
\end{proof}

In the remainder of this paper we shall consider bound states only.
The algebraic formulation of the quantum Kepler problem allows us
to obtain the spectrum of the Hamiltonian for the bound states by
using the representation theory of the algebras described in Lemma
\ref{so(2,1)}. However, in order to determine the multiplicities of
the eigenvalues and also to construct the corresponding
eigenfunctions, we need to explore a larger dynamical symmetry of
the problem. This is done in the next section.

\section{Dynamical supersymmetry}\label{Dynamicalsymmetry}
In this section we shall show that the quantum Kepler problem in the
superspace $\R^{D|2n}$ has a dynamical supersymmetry described by
the orthosymplectic Lie superalgebra $osp(2, D+1|2n)$. The
generators and commutation relations of the superalgebra will be
given explicitly. A parabolic subalgebra and some nilpotent
subalgebra of the dynamical supersymmetry algebra will be studied
in detail, which play crucial roles in solving the quantum Kepler
problem.

\subsection{Dynamical supersymmetry}
We have the following result.
\begin{lemma} \label{Jia}
Let $\Gamma_a= R\partial_a$ for all $a=1, 2, \dots, D+2n$.
\begin{enumerate}
\item
We have $[J_{-2}, \Gamma_a]=0$, and
\begin{eqnarray}
\begin{aligned}
A_a&:=[J_{-1}, \Gamma_a]= \frac{i}{2}X_a(\Delta + 1) - i T\partial_a, \\
M_a&:= [J_{0}, \Gamma_a]= \frac{i}{2}X_a(\Delta - 1) - i T\partial_a.
\end{aligned}
\end{eqnarray}

\item
The operators $\Gamma_{a}$, $A_a$ and $M_a$ satisfy the following commutation relations
\begin{eqnarray} \label{osp-ext}
\begin{aligned}
&{[ \Gamma_a, \Gamma_b ]} = J_{a b}, &&[\Gamma_a, A_b]=-\eta_{b a}J_{-1}, && [\Gamma_a, M_b]=-\eta_{b a}J_{0},\\
& [A_a, A_b] =-J_{a b}, && [M_a, M_b] = J_{a b}, && [A_a, M_b]=-\eta_{b a} J_{-2}.
\end{aligned}
\end{eqnarray}
\end{enumerate}
\end{lemma}
\begin{proof}
The lemma can be proved by straightforward computations. To prove part (1), note that
\[ A_a =-\frac{i}{2}R [\Delta, R] \partial_a + \frac{i}{2}R [\partial_a, R] (\Delta+1).\]
Using the first relation in \eqref{dR} and also \eqref{DeltaR}, we easily arrive at the result.
The formula for $M_a$ can be proved in exactly the same way.

The proof for the first formula in \eqref{osp-ext} is very simple,
thus we omit the details. Now consider $[\Gamma_a, A_b]$. We have
\[
[\Gamma_a, A_b] =\frac{i}{2}\eta_{a b}R -i[R\partial_a, T\partial_b] + \frac{i}{2}[R\partial_a, X_b\Delta].
\]
Using the following relations
\begin{eqnarray}\label{calculations2}
\begin{aligned}
&[R\partial_a, T\partial_b] = - (-1)^{[a][b]}\frac{X_b}{R}T\partial_a, \\
&[R\partial_a, X_b\Delta] = \eta_{a b}R\Delta - (-1)^{[a][b]}\frac{2X_b}{R}T\partial_a,
\end{aligned}
\end{eqnarray}
and also \eqref{DeltaR} we easily arrive at the desired formula. The commutator $[\Gamma_a, M_b]$
can be similarly proved by using relations \eqref{calculations2} and \eqref{DeltaR}.

Let us consider $[A_a, A_b]$. We have
\begin{eqnarray*}
\begin{aligned}
{-[A_a, A_b]} &= [T\partial_a, T\partial_b] -\frac{1}{2} [T\partial_a, X_b(\Delta+1)]\\
 & + (-1)^{[a][b]} \frac{1}{2} [T\partial_b, X_a(\Delta+1)]\\
 &+\frac{1}{4} \left(X_a[\Delta, X_b] - (-1)^{[a][b]}X_b[\Delta, X_a]\right)(\Delta+1).
\end{aligned}
\end{eqnarray*}
To simplify the  right hand side of the above equation, we use the following relations
\begin{eqnarray}\label{calculations1}
\begin{aligned}
&{[T\partial_a, T\partial_b]}=0, \qquad [\Delta, X_a]=2\partial_a, \\
&[T\partial_a, X_b]= T\eta_{b a} +(-1)^{[a][b]}X_b\partial_a, \\
&[T\partial_a, X_b\Delta]= \left(T\eta_{b a} - (-1)^{[a][b]}X_b\partial_a\right)\Delta.
\end{aligned}
\end{eqnarray}
Carefully combining similar terms we arrive at $[A_a, A_b]=-J_{a b}$. The commutators
$[M_a, M_b]$ and $[A_a, M_b]$ can be calculated in the same way by using the relations in \eqref{calculations1}.

\end{proof}

Let us introduce a $(D+2n+3)\times(D+2n+3)$ matrix $\left(\eta_{K L}\right)= \begin{pmatrix}I_{2, 1}&0\\ 0 & \eta\end{pmatrix}$,
where $I_{2, 1}=\begin{pmatrix}-1&0 & 0\\ 0&-1&0\\0&0&1\end{pmatrix}$ and $\eta$ is the metric of $\R^{D|2n}$. The indices $K$ and $L$
take values $-2, -1, 0, 1, 2, \dots, D$ ordered as shown. We set $[K] =0$ if $K=i\le 0$,
and $[K]=[a]$ if $K=a\ge 1$.

\begin{theorem}\label{osp-dynamicalsymmetry}
Introduce the operators $J_{K L}$ such that
$J_{KL}=-(-1)^{[K][L]}J_{L K}$ with
\begin{eqnarray}\label{dynamicalsymmetry}
\begin{aligned}
&J_{i j} = \epsilon_{i j k} J_k, \quad \text{for\ } i, j, k\in\{-2, -1, 0\}, \\
&J_{-2, a} = \Gamma_a, \quad
J_{0, a} = -A_{a}, \quad
J_{-1, a} = M_a,\\
&J_{a b} \ \text{is defined by \eqref{osp-small}},
\quad \text{for\ } a, b\in\{1, 2, \dots, D+2n\},
\end{aligned}
\end{eqnarray}
where $\epsilon_{i j k}$ is totally skew symmetric in the three
indices and $\epsilon_{-2, -1, 0}=1$. These operators form a basis for the
othosymplectic Lie superalgebra $osp(2, D+1|2n)$ with the
commutation relations
\begin{eqnarray}
\begin{aligned}
{[J_{K L}, J_{P Q}]} = & \eta_{P L} J_{K Q} + (-1)^{[K]([L]+[P])}\eta_{Q K} J_{L P}\\
&- (-1)^{[P][Q]}\eta_{Q L} J_{K P} - (-1)^{[K][L]}\eta_{P K} J_{L Q}.
\end{aligned}
\end{eqnarray}
\end{theorem}
\begin{proof}
Lemma \ref{so(2,1)} and Lemma \ref{Jia} have verified all the
commutation relations except $[J_{a b}, J_{i c}]$ for $a, b, c\ge 1$
and $i\le0$, and some cases of $[J_{i j}, J_{k c}]$ with $i, j, k\le 0$
and $c\ge 1$. The relations for $[J_{i j}, J_{k c}]$ easily follow
from Lemma \ref{Jia}(1) and Lemma \ref{so(2,1)}(1). A direct
calculation gives
\[
[J_{a b}, J_{-2, c}] = \eta_{c b} J_{-2, a} - (-1)^{[b][c]} \eta_{c a} J_{-2, b}.
\]
Since $J_{a b}$ commutes with all $J_{i j}$ by
Lemma \ref{so(2,1)}(2), we may replace the index $-2$ in this equation by any $i\in\{-2, -1, 0\}$.
This completes the proof of the theorem.
\end{proof}

\subsection{Root system}
Denote by $\fg$ the complexification of the real Lie superalgebra
$osp(2, D+1|2n)$ spanned by the operators in \eqref{dynamicalsymmetry}, then $\fg\cong osp(D+3|2n)$. Let us now specify
the root system for $\fg$ that will be used for
characterising highest weight representations. We adopt notations
and conventions from \cite{K} (see also \cite{S}), but take a nonstandard root system for
the Lie superalgebra $\fg$. Dynkin diagrams corresponding to the root systems of
 $\fg$ for even and odd $D$ are respectively given by

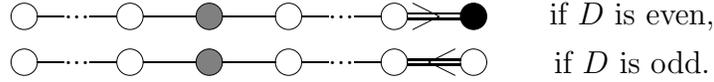
\begin{center}
\begin{picture}(400, 10)(0, 5)
%\begin{picture}(500, 10)(0, 5)

 \put(110, 10){\circle{10}}

\put(115, 10){\line(1, 0){10}}

\put(130, 10){\makebox(0,0){...}}

\put(135, 10){\line(1, 0){10}}

\put(150, 10){\circle{10}}

\put(155, 10){\line(1, 0){20}}

\put(180, 10){\filltype{shade}\circle*{10}}

\put(185, 10){\line(1, 0){20}}

\put(210, 10){\circle{10}}

\put(215, 10){\line(1, 0){10}}

\put(230, 10){\makebox(0,0){...}}

\put(235, 10){\line(1, 0){10}}

\put(250, 10){\circle{10}}

\put(255, 9){\line(1, 0){20}} \put(255, 11){\line(1, 0){20}}

\put(267, 10){\line(-2, 1){10}}
\put(267, 10){\line(-2, -1){10}}

\put(280, 10){\circle*{10}}

\put(340, 10){\makebox(0,0){if $D$ is even,}}
\end{picture}
\end{center}

\begin{center}
\begin{picture}(400, 10)(0, 5)

 \put(110, 10){\circle{10}}

\put(115, 10){\line(1, 0){10}}

\put(130, 10){\makebox(0,0){...}}

\put(135, 10){\line(1, 0){10}}

\put(150, 10){\circle{10}}

\put(155, 10){\line(1, 0){20}}

%\put(175, 9){\line(1, 0){8}}

\put(180, 10){\filltype{shade}\circle*{10}}

\put(185, 10){\line(1, 0){20}}

\put(210, 10){\circle{10}}

\put(215, 10){\line(1, 0){10}}

%\put(225, 10){\makebox(0,0){...}}
\put(230, 10){\makebox(0,0){...}}

\put(235, 10){\line(1, 0){10}}

\put(250, 10){\circle{10}}

\put(263, 10){\line(2, -1){10}}\put(263, 10){\line(2, 1){10}}

\put(255, 9){\line(1, 0){20}}
\put(255, 11){\line(1, 0){20}}

\put(280, 10){\circle{10}}

\put(340, 10){\makebox(0,0){if $D$ is odd.}}

\end{picture}
\end{center}

\bigskip
\noindent Here the black and grey nodes correspond to odd simple roots.

When $D=2m$, there are $m+1+n$ nodes in the Dynkin diagram. Order
the nodes from left to right. They respectively correspond to the
simple roots
\[
\epsilon_0-\epsilon_1,  \epsilon_1-\epsilon_2, \dots,
\epsilon_{m-1}-\epsilon_m,  \epsilon_m-\delta_1,  \delta_1-\delta_2,
\dots,  \delta_{n-1}-\delta_n,  \delta_n.
\]
The element $h_0$ belongs to the Cartan subalgebra with
$\epsilon_0(h_0)=1$ and $\epsilon_i(h_0)=0=$ $\delta_j(h_0)$ for all
the other $\epsilon_i$ and all the $\delta_j$.

When $D=2m+1$, there are $m+2+n$ nodes, respectively corresponding
to the simple roots
\[
\epsilon_{-1}-\epsilon_0, \epsilon_0-\epsilon_1,
\epsilon_1-\epsilon_2, \dots, \epsilon_{m-1}-\epsilon_m,
\epsilon_m-\delta_1,  \delta_1-\delta_2, \dots,
\delta_{n-1}-\delta_n,  2\delta_n.
\]
In this case, $\epsilon_{-1}(h_0)=1$ while the evaluations of all
the other $\epsilon_i$ and all the $\delta_j$ on $h_0$ are zero.

An element in the weight space of $\fg$ will be written as an
$\left[\frac{D+3}{2}\right]+n$ tuple, which is the coordinate in the
basis consisting of the $\epsilon_i$ and
$\delta_j$. The basis is ordered in such a way that the
$\delta_j$ are positioned after all $\epsilon_i$,
and the $\epsilon_i$
appear in their natural order, and so do also the
$\delta_j$.

Since $osp(D+1|2n)$ is a regular subalgebra of $\fg$, we shall take
for it the root system compatible with that of $\fg$. This is
specified by the Dynkin diagram obtained by removing the left most
node from the Dynkin diagram of $\fg$.

\subsection{Parabolic subalgebra} In this subsection we discuss
subalgebras of the Lie superalgebra $\fg$. The main result
established is Proposition \ref{irred}, which is of crucial
importance for solving the Kepler problem on the superspace.
Unfortunately the material presented here is quite technical, so we have relegated some of
it to the Appendix.

Define the linear operator $ad_{h_0}$ on
$\fg$ by $ad_{h_0}(Y) = [h_0,  Y]$ for all $Y\in \fg$ (recall that $h_0=iJ_0$).
Now $\fg$ decomposes into a direct sum
$\fg = \fg_{-1} \oplus \fg_0 \oplus \fg_{+1}$
of eigenspaces of $ad_{h_0}$ with eigenvalues $+1$, $0$ and $-1$ respectively. Here
\begin{enumerate}
\item $\fg_{+1}$ is spanned by $i(J_{-1} - i J_{-2})$ and $M_a - i\Gamma_a$ ($a\ge 1$),
\item $\fg_{-1}$ is spanned by $i(J_{-1} + i J_{-2})$ and $M_a + i\Gamma_a$ ($a\ge 1$),
\item $\fg_{0}$ is spanned by $h_0$, $A_a$ and $J_{a b}$ ($a, b\ge 1$).
\end{enumerate}
The subspaces $\fg_{+1}$, $\fg_{-1}$ and $\fg_0$ all form
subalgebras of $\fg$. In particular, $\fg_0$ is the subalgebra
$gl_1\oplus osp(D+1|2n)$, and  $[\fg_{+1}, \fg_{+1}]= [\fg_{-1},
\fg_{-1}]=0$. Note that if $D$ is even, $\fg_{+1}$ is spanned by the
root vectors associated with the positive roots
\[
\epsilon_0, \epsilon_0\pm \epsilon_i, \forall i>0; \ \ \epsilon_0\pm
\delta_j, \forall j.
\]
If $D$ is odd, the subalgebra $\fg_{+1}$ is spanned by the root
vectors associated with the positive roots
\[\epsilon_{-1}\pm \epsilon_i, \forall i\ge 0; \ \ \epsilon_{-1}\pm \delta_j,
\forall j.
\]

Both $\fg_{+1}$ and $\fg_{-1}$ are stable under the adjoint action
of $\fg_0$ in the sense that $[\fg_0, \fg_{\pm 1}]= \fg_{\pm 1}$, thus form $\fg_0$-modules.
When restricted to modules for the $osp(D+1|2n)$ subalgebra of
$\fg_0$, both $\fg_{+1}$ and $\fg_{-1}$ are isomorphic to the natural
$osp(D+1|2n)$-module.

We have the following parabolic subalgebra of $\fg$:
\begin{eqnarray}\label{p}
\fp:= \fg_0 \oplus \fg_{+1}
\end{eqnarray}
with the nilpotent radical $\fg_{+1}$. Furthermore, $\fg=\fp\oplus\fg_{-1}$.

Denote by $\cU_\fg$ the complex linear span of products of the
operators in $\fg$.  Let $\cU_\fp$ and $\cU_{\fg_{-1}}$ be the sub superalgebras of $\cU_\fg$
generated by the elements of $\fp$ and $\fg_{-1}$ respectively. Then
the underlying vector space of $\cU_\fg$ is isomorphic to
$\cU_{\fg_{-1}}\cU_\fp$. Note that  $\cU_{\fg_{-1}}$ is $\Z_2$-graded
commutative.

\begin{remark}
The associative superalgebra $\cU_{\fg}$ is a quotient of
the universal enveloping algebra $\U(\fg)$ of $\fg=osp(D+3|2n)$.
Similar comments apply to $\cU_\fp$ and $\U(\fp)$ etc..
\end{remark}

The subalgebra $\cU_{\fg_{-1}}$ has a $\Z$-grading $\cU_{\fg_{-1}} =
\oplus_{l=0}^\infty\cU_{\fg_{-1}}(-l)$.
Let $K_0=i(J_{-1} + i J_{-2})$ and $K_a=M_a + i\Gamma_a$ ($a\ge 1$),
which form a basis of the subalgebra $\fg_{-1}$.
Then $\cU_{\fg_{-1}}(-l)$
is spanned by all the homogeneous polynomials of order $l$ in the
elements $K_A$ with $A=0, 1, \dots, D+2n$. Here we should note that
$K_A$ are odd in the $\Z_2$ grading for all $A>D$, and in this case
$(K_A)^2=0$. However, $K_0$ is even and $(K_0)^l\ne 0$ for all $l$.

Let $\eta_0=(\eta_{A B})$ with $A, B=0, 1, \dots, D+2n$ be the matrix
obtained from $(\eta_{K L})$ by removing the first two rows and
first two columns. This is the metric relative to which the $osp(D+1|2n)$
subalgebra of $\fg_0$ is defined. Now the (adjoint) action of $\fg_0$ on $\fg_{-1}$ naturally extends to
an action on $\cU_{\fg_{-1}}(-l)$.  Denote by $\eta_0(\fg_{-1}, \fg_{-1})$
the subspace of $osp(D+1|2n)$-invariants in $\cU_{\fg_{-1}}(-2)$. Then $\eta_0(\fg_{-1}, \fg_{-1})$ is
spanned by
\[ (K)^2 :=\sum_{A, B=0}^{D+2n} \eta^{B A}K_A K_B, \]
where $\eta^{A B}$ are the matrix elements of $\eta_0^{-1}$. We have the following result.

\begin{lemma}\label{traceless}
The operators $K_A$ satisfy $(K)^2=0$.
\end{lemma}
\begin{proof}
The proof of this claim involves a considerable amount of
calculations. Thus we shall present only the main formulae needed.
Let us first calculate $K_0^2$. We have
\begin{eqnarray*}
\begin{aligned}
K_0^2 &= \left(\frac{1}{2}R(\Delta+1)-T\right)^2 \\
&=\frac{1}{4} \left(R(\Delta+1)\right)^2-T R (\Delta+1) -
\frac{1}{2}[R(\Delta+1), T] + T^2.
\end{aligned}
\end{eqnarray*}
Using $[R(\Delta+1), T]=R(\Delta-1)$, we obtain
\begin{eqnarray}\label{K02}
K_0^2 &=&\frac{1}{4} \left(R(\Delta+1)\right)^2 - T R (\Delta+1) -
\frac{1}{2}R(\Delta-1) + T^2.
\end{eqnarray}

Now $\sum_{a=1}^{D+2n} K^aK_a = \sum_a ( M^a M_a + i(M^a \Gamma_a
+\Gamma^a M_a)- \Gamma^a \Gamma_a)$.  Note that
\[
\begin{aligned}
\sum_{a=1}^{D+2n} M^a M_a &= -\frac{1}{4}\sum_a
X^a(\Delta-1)X_a(\Delta-1) \\
&+ \frac{1}{2}\sum_a (X^a(\Delta-1) T\partial_a + T\partial^a
X_a(\Delta-1)) \\&- \sum_a T\partial^a T\partial_a.
\end{aligned}
\]
The first term can be expressed as $ - \frac{1}{4}R^2 (\Delta-1)^2 -
\frac{1}{4}\sum_a X^a[(\Delta-1), X_a](\Delta-1)$.  Using
$[(\Delta-1), X_a]= 2\partial_a$ we obtain
$ - \frac{1}{4}R^2 (\Delta-1)^2 -\frac{1}{2}E (\Delta-1). $
The second term can be simplified to
$ T(T+1/2)(\Delta-1) + \frac{1}{2}E(\Delta +1), $
while the third yields $-T(T+1)\Delta$. Combining these results
together we arrive at
\begin{eqnarray}\label{MM}
\sum_{a=1}^{D+2n} M^a M_a = - \frac{1}{4}R^2 (\Delta-1)^2
-\frac{1}{2}T(\Delta+1) - T^2+E.
\end{eqnarray}
It is easy to show that
\begin{eqnarray}\label{GammaGamma}
\sum_{a=1}^{D+2n} \Gamma^a \Gamma_a = E + R^2\Delta.
\end{eqnarray}
Now $2i\sum_{a=1}^{D+2n} (M^a \Gamma_a +\Gamma^a M_a)$ can be expressed as
\begin{eqnarray}\label{calculations}
\begin{aligned}
-\sum_{a}X^a(\Delta-1)R\partial_a -\sum_{a}R\partial^a X_a(\Delta-1)
+ 2\sum_a(T\partial^a R\partial_a + R\partial^a T\partial_a).
\end{aligned}
\end{eqnarray}
The first sum  can be rewritten as
$ \sum_{a} X^a[(\Delta-1), R]\partial_a + RE(\Delta-1). $
Using equation \eqref{DeltaR} we obtain $\sum_{a} X^a[(\Delta-1),
R]\partial_a= \sum_a\frac{2X^a}{R} T\partial_a= 2T\frac{1}{R}E$.
This leads to
\[
\sum_{a}X^a(\Delta-1)R\partial_a = 2T\frac{1}{R}E + RE(\Delta-1).
\]
We also have
\[
\sum_{a}R\partial^a X_a(\Delta-1)= R(d + E)(\Delta-1).
\]
Finally the third sum in \eqref{calculations} gives
\[
2\sum_a(T\partial^a R\partial_a + R\partial^a T\partial_a) =
2T\frac{1}{R}E +4TR\Delta.
\]
Combining the results together we obtain
\begin{eqnarray}\label{MGamma}
\begin{aligned}
i\sum_{a=1}^{D+2n} (M^a \Gamma_a +\Gamma^a M_a)
=\frac{1}{2}R(\Delta-1) + TR(\Delta+1).
\end{aligned}
\end{eqnarray}

Combining equations \eqref{MM}, \eqref{GammaGamma} and
\eqref{MGamma} we readily obtain
\begin{eqnarray*}
\begin{aligned}
\sum_{a=1}^{D+2n} K^a K_a = &- \frac{1}{4}R^2 (\Delta-1)^2-R^2\Delta
-\frac{1}{2}T(\Delta+1) \\
& + TR(\Delta+1)+\frac{1}{2}R(\Delta-1) - T^2.
\end{aligned}
\end{eqnarray*}
Using $R^2(\Delta+1)^2=(R(\Delta+1))^2 -2 T(\Delta+1)$, we
immediately arrive at
\begin{eqnarray}
\begin{aligned}
\sum_{a=1}^{D+2n} K^a K_a = - \frac{1}{4}(R(\Delta+1))^2 +
TR(\Delta+1)+\frac{1}{2}R(\Delta-1) - T^2.
\end{aligned}
\end{eqnarray}
Recalling equation \eqref{K02}, we see that
$K_0^2+ \sum_{a=1}^{D+2n} K^a K_a=0.$
This completes the proof.
\end{proof}

With the help of Lemma \ref{traceless}, we can prove the following
result, which will play a crucial role in understanding the $osp(2,
D+1|2n)$ representation appearing in the quantum Kepler problem.

\begin{proposition}\label{irred}
As a module for the subalgebra $osp(D+1|2n)$ of $\fg_0$, the
subspace $\cU_{\fg_{-1}}(-l)$ of $\cU_{\fg_{-1}}$ is irreducible and
isomorphic to the irreducible rank $l$ symmetric (in the
$\Z_2$-graded sense) tensor of the natural module for $osp(D+1|2n)$.
The $osp(D+1|2n)$ highest weight of $\cU_{\fg_{-1}}(-l)$ is given by $(l, 0,
\dots, 0)$.
\end{proposition}

\begin{remark}
The analogous statement for the subalgebra $\U(\fg_{-1})$ of the
universal enveloping algebra $\U(\fg)$ of $\fg=osp(D+3|2n)$ is
completely false.
\end{remark}

\begin{proof}[Proof for Proposition \ref{irred}]
We need the general facts about symmetric tensors of the natural
module $V$ for $osp(D+1|2n)$ discussed in Appendix \ref{appendix}.
Here it is more convenient to choose a basis $\{v^A|A=0, 1, \dots,
D+2n\}$ for $V$ such that the non-degenerate invariant bilinear form
$\langle\ , \ \rangle: V\times V\longrightarrow \C$ gives the metric
$(\eta^{A B})$. Now we use this matrix to define $\square$ and
$\square^*$.  Under the condition that $D>2n+1$, every symmetric
tensor power $S(V)_l$ of $V$ is semisimple as a $osp(D+1|2n)$-module
by Lemma \ref{key}. Furthermore, $S(V)_l=S(V)_l^0 \oplus
S(V)_{l-2}\square^*$, and the harmonic space $S(V)_l^0$ defined by
$S(V)_l^0=\{w\in S(V)_l \mid \square w =0\}$ is irreducible with
highest weight $(l, 0, \dots, 0)$.

Now we turn to $\cU_{\fg_{-1}}$. The $\Z$-graded superalgebra
homomorphism $S(V)\to \cU_{\fg_{-1}}$ is also an $osp(D+1|2n)$ map.
The restriction of this map to any homogeneous component is nonzero.
Since $(K)^2=0$ by Lemma \ref{traceless}, the map sends
$S(V)_{l-2}\square^*$ to zero. Thus $\cU_{\fg_{-1}}(-l)$ is the
image of $S(V)_l^0$. Since $S(V)_l^0$ is irreducible, we must have
$\cU_{\fg_{-1}}(-l)\cong S(V)_l^0$ as $osp(D+1|2n)$-modules.
\end{proof}

\section{Solution of quantum Kepler problem on superspace}

\subsection{Induced representations}

Let $\Phi_0 = g_0 \Psi_0$ with $g_0=\exp\left(-T \ln\sqrt{-2\cE_0} \right)$,
where $\Psi_0$ is the ground state wave function.
Then $\Phi_0= (-2\cE_0)^{\frac{d-1}{4}}e^{-R}$.
\begin{lemma} \label{hw-module}
The function $\Phi_0$ spans a $1$-dimensional module for the
parabolic subalgebra $\fp$ defined by equation \eqref{p}.
\end{lemma}
\begin{proof}
It is evident that $J_{a b}(\Phi_0)=0$ for all $a, b\ge 1$. Also the definition of $\Phi_0$ implies that it is an eigenvector
of $h_0$. Therefore the lemma will be valid if $\Phi_0$ is annihilated by
$i(J_{-1} - i J_{-2})$, $A_a$ and $M_a - i \Gamma_a$ ($a\ge 1$).

To proceed further, we need to use the following formulae:
\begin{eqnarray} \label{nablaPh0}
\begin{aligned}
&\Delta(e^{-R}) = \left(1-\frac{d-1}{R}\right)e^{-R}, \\
&T(e^{-R}) =\left(\frac{d-1}{2}-R\right)e^{-R},\\
&T\partial_a(e^{-R}) = X_a \left(1-\frac{d-1}{2R}\right)e^{-R}.
\end{aligned}
\end{eqnarray}
Using the first and third formulae in the following equation
\[ A_a (e^{-R}) = \frac{i}{2}X_a(\Delta + 1)(e^{-R}) -  iT\partial_a(e^{-R}),  \]
we immediately arrive at $A_a (e^{-R})=0$. Note that
\[M_a (e^{-R}) = A_a (e^{-R}) - i X_a e^{-R} =  - i X_a e^{-R}.\]
This easily leads to
\[ (M_a - i \Gamma_a)(e^{-R}) =0, \quad \forall a\ge 1. \]
Finally by using the first and second formulae of \eqref{nablaPh0},
we can show that \[i(J_{-1} - i J_{-2})(e^{-R})=0.\]
\end{proof}

Let $\cH=\cU_\fg\Phi_0$, which has a natural $\fg$-module structure.
Then
\[
\cH=\cU_{\fg_{-1}}\Phi_0 = \oplus_{l=0}^\infty\cU_{\fg_{-1}}(-l)\Phi_0.
\]
It follows that the operator $h_0$ is diagonalisable in $\cH$ and its eigenvalues
must be of the form $-\frac{d-1}{2}-k$ ($k\in\Z_+)$.

If $L^0_\lambda$ denotes the irreducible
$\fp$-module with highest weight $\lambda$, we have the generalised Verma module
$V_{\lambda} = \U(\fg)\otimes_{\U(\fp)}L^0_\lambda$. The $\fp$-module $L^0_\lambda$ is $1$-dimensional if
\begin{eqnarray}\label{hw}
\lambda = \left(-\frac{d-1}{2}, 0, \dots, 0\right),
\end{eqnarray}
and it is evident that $\cH$ is a quotient module of the
corresponding generalised Verma module $V_{\lambda}$.

The highest weight vector $\Phi_0$ of $\cH$ generates a module for the $so(2, D+1)$ subalgebra of $osp(2, D+1|2n)$
with a highest weight $\lambda^0= \left(-\frac{d-1}{2}, 0, \dots, 0\right)$ (the first $\left[\frac{D+3}{2}\right]$
entries of \eqref{hw}). This module is necessarily infinite dimensional as the highest weight is not dominant.
This in particular implies that
\[\cU_{\fg_{-1}}(-l)\Phi_0\ne 0, \quad \text{for all $l$}.
\]

Another fact which can be deduced is that if $\cH$ contains any nontrivial submodule $\cH_1$,
then $\cH/\cH_1$ has to be infinite dimensional as well, as the image of $\Phi_0$ generates
an infinite dimensional $so(2, D+1)$-submodule in the quotient.

Since $\Phi_0$ has trivial $\fg_0$-action and $\cU_{\fg_{-1}}(-l)$
is irreducible under the action of $\fg_0$ by Proposition
\ref{irred}, we immediately see that $\cU_{\fg_{-1}}(-l)\Phi_0$ is
irreducible as a $\fg_0$-module. Now let $\cH_1$ be a  non-trivial
$\fg$-submodule of $\cH$, then there exists some integer $l_0$ such
that all $\cU_{\fg_{-1}}(-l)$ with $l\ge l_0$ belong to $\cH_1$.
This contradicts the fact that $\cH/\cH_1$ is infinite dimensional.
Thus $\cH$ must be irreducible as a $\fg$-module.

Let $L_\mu^{D+1|2n}$ denote the irreducible $osp(D+1|2n)$-module with highest weight $\mu$.
We have proved the following theorem.
\begin{theorem}\label{main-rep}
\begin{enumerate}
\item As an $osp(2, D+1|2n)$-module $\cH$ is irreducible.
\item The restriction of $\cH$ is
isomorphic to $\oplus_{l=0}^\infty L_{(l, 0, \dots, 0)}^{D+1|2n}$ as
an $osp(D+1|2n)$-module.
\end{enumerate}
\end{theorem}

\begin{remark}
It is known \cite{J, EHW} that $\lambda^0$ does not give rise to a
unitarisable highest weight $so(2, D+1)$-module. Thus the $so(2,
D+1)$-submodule of $\cH$ generated by $\Phi_0$ is not unitarisable,
and this in turn implies that $\cH$ as an $osp(2, D+1|2n)$-module is
not unitarisable.
\end{remark}

\subsection{Solution of quantum Kepler problem on superspace}
In this subsection we use the results on the irreducible $osp(D+3|2n)$-module
$\cH$ obtained to determine the bound states of the
quantum Kepler problem on the superspace $\R^{D|2n}$. Let $\cH_k=\cU_{\fg_{-1}}(-k)\Phi_0$,
and for $k\in\Z_+$,
\[ \cE_k = -\frac{1}{2}\left(\frac{1}{\frac{d-1}{2}+k}\right)^2.\]
Define $g_k=\exp\left(-T \ln\sqrt{-2\cE_k} \right)$, and set
\[
\quad \tilde{\cH}=\bigoplus_{k=0}^\infty \tilde{\cH}_k, \quad \text{where\ \ } \ \tilde{\cH}_k = \{ g^{-1}_k v \mid v\in \cH_k\}.
\]
We have the following result.
\begin{theorem}\label{main-spec}
\begin{enumerate}
\item The quantum Hamiltonian operator $H$ is diagonalisable when acting on $\tilde\cH$ and has the eigenvalues
$\cE_l$ for $l\in\Z_+$.
\item The dimension of the subspace $\tilde\cH_l$ is given by the formula
\[
\sum_{k=0}^l\begin{pmatrix}D+k\\ \\k\end{pmatrix}
\left(\begin{pmatrix}2n\\  \\ l-k\end{pmatrix} - \begin{pmatrix}2n\\
\\ l-2-k\end{pmatrix}\right),
\]
where the binomial coefficient $\begin{pmatrix}a\\ \\b\end{pmatrix}$
is assumed to be zero if $b<0$ or $b>a$.
\item The subspace $\tilde\cH_I$ ($I=0, 1, \dots$) is the entire $\cE_I$-eigenspace of the
quantum Hamiltonian operator $H$.
\end{enumerate}
\end{theorem}
\begin{proof}
It follows from Theorem \ref{formulation} that ever nonzero vector in $\tilde\cH_l$ is indeed
an eigenvector of the quantum Hamiltonian operator $H$ with eigenvalue $\cE_l$.
To prove the formula for the dimension of $\tilde\cH_l$, we note that
$\dim\tilde\cH_l=\dim \cH_l$, and by Theorem \ref{main-rep},
$\dim\tilde\cH_l=\dim L_{(l, 0, \dots, 0)}^{D+1|2n}$. We have
\[\dim L_{(k, 0, \dots, 0)}^{D+1|2n} = \dim S(V)_l - \dim S(V)_{l-2}.\]
It is well known that
\[\dim S(V)_l=\sum_{k=0}^l\begin{pmatrix}D+k\\ \\k\end{pmatrix}
\begin{pmatrix}2n\\  \\ l-k\end{pmatrix}. \]
Hence the formula for $\dim\tilde\cH_l$ follows.

To prove the third claim, we need some input from the Schr\"{o}dinger equation.
Let us return to the original form \eqref{Schroedinger} of the
equation. Because of the $osp(D|2n)$ symmetry of the
equation, the wave functions $\Psi$ can be written as $\C$-linear combinations of functions
of the form $\frac{\omega}{R^l}\chi_l$ by the first part of Lemma \ref{key}, where $\omega$ is a
harmonic polynomial (that is, $\Delta \omega =0$) in the coordinates
$X^a$, which is homogeneous of degree $l$. The function $\chi_l$
depends on $R$ only, and the factor $R^{-l}$ is introduced for
convenience. Then equation \eqref{Schroedinger} reduces to the
following equation for $\chi_l$
\begin{eqnarray}\label{radial-part}
-\frac{1}{2}\left( \frac{d^2 \chi_l}{d R^2} + \frac{d-1}{R}\frac{d
\chi_l}{d R}  - \frac{l(d-2+l)}{R^2}\chi_l\right)-\frac{1}{R}\chi_l
&=&\cE\chi_l.
\end{eqnarray}
This has the same form as the radial part of the Schr\"{o}dinger
equation for the Kepler problem on $\R^d$, and can be solved in
terms of generalised Laguerre polynomials (see, e.g., \cite[(4),
(14)]{Al}). Such a solution, when its argument $R$ is replaced by
$r$, is square integrable over the positive half line $\R_+$ with
respect to the measure $r^{d-1} dr$. For a fixed $l$, we denote by
$\chi_{l, j}$ ($j=0, 1, 2, \dots$) the independent solutions of
\eqref{radial-part}.

Consider $\langle \frac{\omega}{R^l}\chi_{l, j}\mid
\frac{\omega}{R^l}\chi_{l, j}\rangle =\int_{\R^D}\bar{\omega}\omega
\left(\frac{\chi_{l, j}}{R^l}\right)^2$, the integrant of which
vanishes exponentially fast as $r$ goes to infinity. Thus we only
need to analyse the $r\rightarrow 0$ end of the integral to see
whether the integral converges. By inspecting the form of the
generalised Laguerre polynomials we can see that at the worst, the
integral behaves like $\int_{\R^D}
\left(\frac{\Theta^2}{R^2}\right)^l
\left(\frac{\Theta^2}{R^2}\right)^{n-l} (\chi_{l, j}(R))^2\propto (\Theta^2)^n\int_0^\infty  ({\chi_{l,
j}(r)})^2 r^{d-1} dr$. It follows from the square integrability of
$\chi_{l, j}$ over $\R_+$ with respect to the measure $r^{d-1} dr$
that all $\frac{\omega}{R^l}\chi_{l, j}$ belong to the set $\cF$
defined in subsection \ref{Kepler}.

For a fixed $l$, the solutions $\chi_{l, j}$ of \eqref{radial-part}
respectively correspond, in a one-to-one manner, to the eigenvalues
$\cE_{l+j}$ of the quantum Hamiltonian operator \eqref{Hamiltonian}.
In particular, the ground state corresponds to the energy eigenvalue
$\cE_0$. Thus for a given non-negative integer $I$, the eigenspace
of $H$ corresponding to the energy $\cE_I$ is spanned by
$\frac{\omega}{R^l}\chi_{l, j}$ for all homogeneous harmonic
polynomials $\omega$ of degree $l$, and for all $l, j\ge 0$ with
$l+j=I$. As a module over $osp(D|2n)$, the $\cE_I$-eigenspace is
isomorphic to $\bigoplus_{l=0}^I L^{D|2n}_{(l, 0, 0, \dots, 0)}$. By
using the branching rule \eqref{branching}, we have
\[ \bigoplus_{l=0}^I L^{D|2n}_{(l, 0, 0, \dots, 0)} \cong L^{D+1|2n}_{(I, 0, 0, \dots, 0)}
\cong \tilde\cH_I.
\]

\end{proof}

\begin{remark}\label{comments2}
In view of Remark \ref{comments1}, the
generalised Kepler problem on the superspace is equivalent to an eigenvaule problem
involving a system of partial differential equations.
It is a rather non-trivial matter that the problem is still soluble.
\end{remark}

\begin{appendix}
\section{Symmetric superalgebra}\label{appendix}

For proving Theorem \ref{irred}, we need some general facts about
symmetric tensors of the natural module $V$ for $osp(M|2n)$,
which we briefly discuss in this appendix.
In general the symmetric tensors are not completely reducible,
and this complicates the matter of decomposing these representations
enormously. Fortunately  we only need to
consider the case $M-2n>1$ for the purpose of this paper. In this case,
the symmetric tensor are semi-simple, as we shall show.

We work over the complex field and use the root system described in
Section \ref{Dynamicalsymmetry} for $osp(M|2n)$. For the sake of
being concrete, we label the basis elements $\epsilon_i$ and
$\delta_j$ of the weight space by $i=1, 2, \dots,
\left[\frac{M}{2}\right]$ and $j=1, 2, \dots, n$.  Let $v^a$ with
$a=1, 2, \dots, M+2n$ be a weight basis for the natural module $V$,
where the weight $wt(v^a)$ of $v^a$ is greater than that of $v^b$ if $a<b$.
Then the highest weight is $\epsilon_1$, and the lowest weight is
$-\epsilon_1$. We assume that the highest weight vector $v^1$ is
even.

There exists a non-degenerate $osp(M|2n)$-invariant bilinear
form $\langle \ , \ \rangle: V\times V\longrightarrow \C$, which is
unique up to scalar multiples. Note that $\langle v^a, v^b\rangle \ne 0$ if and only if
 $wt(v^a)$ $=-wt(v^b)$. Let $\eta^{a b}=\langle v^a,
v^b\rangle$ and  form the matrix $\eta^{-1} = (\eta^{a b})$. We also
write $\eta=(\eta_{a b})$ for the inverse matrix.

Let $S(V)$ be the $\Z_2$-graded symmetric superalgebra of $V$, that
is, the superalgebra generated by the $v^a$ subject to the relations
that  $v^a v^b = - v^b v^a$ if both $v^a$ and $v^b$ are odd, and
$v^a v^b = v^b v^a$ otherwise. Then this is a $\Z$-graded algebra
$S(V)=\oplus_{l=0}^\infty S(V)_l$ with elements of $V$ having degree
$1$.

Define operators
\[ \square^* = \sum_{a=1}^{M+2n} v^a\eta_{a b} v^b, \quad \square
= \sum_{a=1}^{M+2n} \eta^{b a} \frac{\partial}{\partial{v^a}}
\frac{\partial}{\partial{v^b}}, \quad T=
\frac{M+2n-1}{2}+\sum_{a=1}^{M+2n}
v^a\frac{\partial}{\partial{v^a}},
\]
which all commute with $osp(M|2n)$. The operators satisfy the
commutation relations
\[
[T,  \square^* ]=2 \square^*, \quad [T,  \square ]=-2 \square, \quad
[\square^*, \square]=-T,
\]
thus their real spanned  is isomorphic to the Lie algebra $su(1,
1)$.

We consider $S(V)$ as a complex module for $su(1, 1)$.  The operator
$T$ is diagonalisable with eigenvalues $\frac{M+2n-1}{2}+l$ for
$l\in\Z_+$. Now we need to assume that
\begin{eqnarray}\label{dim-cond} M-2n >1. \end{eqnarray}
Evidently every submodule of $S(V)$ is of highest weight type. Since
the eigenvalues of $-T$ are all strictly negative, the
sub-representations are necessarily infinite dimensional.
The eigenvalues of the quadratic Casimir of $su(1, 1)$
corresponding to distinct highest weights $\lambda_l:=-\left(\frac{M+2n-1}{2}+l\right)$ with $l\in\Z_+$ are different,
thus $S(V)$ decomposes into a direct sum $S(V)=\oplus_l C^{(l)}$, where each submodule $C^{(l)}$ has the property that
all is irreducible subquotients are isomorphic with the same highest weight $\lambda_l$.
An irreducible $su(1, 1)$-module with highest weight $\lambda_l$ for
$l\in\Z_+$ is unitarisable. It follows that every isotypical component is unitarisable
and hence completely reducible. Thus $S(V)$ is completely reducible
with respect to $su(1, 1)$.

This in particular implies that every weight vector
in $S(V)_l$ can be uniquely expressed as $v+w$ with $v\in S(V)_l\cap
ker\square$ being a highest weight vector, and $w\in \square^*\square S(V)_l$.
We can also deduce that $\square^*\square S(V)_l=S(V)_l\cap im\square^*$ for each
$S(V)_l$ by noting the obvious fact that if $v$ belongs to a weight
space in an irreducible $su(1, 1)$-module, then $\square^*\square v$
is a nonzero scalar multiple of $v$ unless $v$ is the highest weight
vector. Thus we have the following vector space decomposition
\begin{eqnarray}\label{sl2-decomposition}
S(V) = ker\square \oplus im\square^*.
\end{eqnarray}
Since the $su(1, 1)$ algebra commutes with $osp(M|2n)$, equation
\eqref{sl2-decomposition} is a decomposition of $osp(M|2n)$-modules.

Denote $S(V)_l\cap ker\square$ by $S(V)_l^0$ and call it the
harmonic space of $S(V)_l$.  It is easy to see that  for all $l\le
2$,  $S(V)_l^0$ is isomorphic to the irreducible $osp(M|2n)$-module
with highest weight  $(l, 0, \dots, 0)$. Assume that this is also
true for $l>2$, then the highest weight vector for $S(V)_l^0$ is
$(v^1)^l$. Each highest weight vector in $S(V)_{l+1}^0$ must contain
a term $(v^1)^l v^b$ for some $b$. In order for the corresponding
weight to be dominant, $v^b$ is either $v^1$, $v^2$ or the lowest
weight vector of $V$, which respectively have weights $(l+1, 0,
\dots, 0)$, $(l, 1, 0, \dots, 0)$ and $(l-1, 0, \dots, 0)$. We can
write down all the vectors of the same weights in $S(V)_{l+1}$, and
try to make linear combinations of them to obtain highest weight
vectors. Simple calculations show that there can not be any highest
weight vector with weight $(l, 1, 0, \dots, 0)$ as $S(V)$ is the
symmetric superalgebra. The highest weight vector corresponding to
the weight $(l-1, 0, \dots, 0)$ is $(v^1)^{l-1}\square^*$,  which
belongs to $im\square^*$ but not $S(V)_{l+1}^0$. This proves that
$S(V)_{l+1}^0$ is isomorphic to the irreducible $osp(M|2n)$-module
with highest weight $(l+1, 0, \dots, 0)$.

To summarise, we have the following result.
\begin{lemma}\label{key}
Keep notations as above.
\begin{enumerate}
\item Under the condition
\eqref{dim-cond}, $S(V)_l = S(V)_l^0 \oplus S(V)_{l-2} \square^*$ as
$osp(M|2n)$-module, and $S(V)_l^0$ is isomorphic to the irreducible
module $L_{(l, 0, \dots, 0)}^{M|2n}$ with highest weight $(l, 0,
\dots, 0)$.
\item As an $su(1, 1)\times osp(M|2n)$-module, $S(V) \cong \bigoplus_{L=0}^\infty L^{(l)}\otimes L_{(l, 0, \dots,
0)}^{M|2n}$ where $L^{(l)}$ is the irreducible $su(1,1)$-module with
highest weight $-\frac{M+2n-1}{2}-l$.
\end{enumerate}
\end{lemma}

A further result which can be deduced from the above lemma is the
branching rule of the irreducible symmetric tensor module $L_{(l, 0,
\dots, 0)}^{M|2n}$ to $osp(M-1|2n)$-modules. We assume that the
condition $M-1-2n>1$ is satisfied. Then the symmetric tensor powers
of the natural module $V'$ for $osp(M-1|2n)$ is completely reducible
by the lemma. We have the following $osp(M-1|2n)$-module isomorphism
$S(V)_l \cong \bigoplus_{k=0}^l S(V')_{l-k}$. Then by the first part
of lemma \ref{key},
\[S(V)_l \cong
 \left(\oplus_{k=0}^l S(V')^0_{l-k}\right) \bigoplus \left(\oplus_{k=0}^{l-2}S(V')_{l-k}\right),  \]
where $S(V')^0_i$ is the harmonic subspace of $S(V')_i$.
Using the first part of lemma \ref{key} to the left hand side, and also noting that
the second term on the right hand side can be re-written as
 $\oplus_{k=0}^{l-2}S(V')_{l-k}\cong S(V)_{l-2}$, we obtain
\[
S(V)^0_l \bigoplus S(V)_{l-2} \cong \left(\oplus_{k=0}^l
S(V')^0_{l-k}\right)\bigoplus S(V)_{l-2}.
\]
 That is, the following branching rule holds:
\begin{eqnarray}\label{branching}
L_{(l, 0, \dots, 0)}^{M|2n} \cong \bigoplus_{k=0}^l L_{(l-k,
0, \dots, 0)}^{M-1|2n} \quad \text{as $osp(M-1|2n)$-module}.
\end{eqnarray}

\end{appendix}

\vspace{1cm}

\noindent{\bf Acknowledgement}. I thank Guowu Meng for helpful correspondences. Financial support from
the Australian Research Council is gratefully acknowledged.

%\vspace{1cm}

\end{document}